\documentclass[3p, preprint]{elsarticle}

  \usepackage{etex}
  \usepackage{verbatim}
  \usepackage{amsfonts,amsmath,amsthm}
  \usepackage{tikz}
  \usepackage{array}
  \usepackage{xcolor}
  \usepackage{url}
  \usepackage{lineno}
  \usepackage[whileod]{chl-alg}

  \usepackage[algo2e, noend, noline, boxed]{algorithm2e}
  \SetKwComment{tcm}{\{}{\}}

\newsavebox{\mybox}
\newenvironment{dsproblem}[1]
{\smallskip\centering\begin{lrbox}{\mybox}\begin{minipage}{0.977\columnwidth}#1\\}
{\end{minipage}\end{lrbox}\fbox{\usebox{\mybox}}\smallskip}

  \newcommand{\defproblem}[3]{\begin{dsproblem}{#1}\textbf{Input:} #2\\\textbf{Output:} #3\end{dsproblem}}

\newcommand{\F}{\mathcal{F}}
\newcommand{\MAW}{\textsc{MAW-SequenceComparison}}
\newcommand{\CMAW}{\textsc{MAW-CircularSequenceComparison}}
\newcommand{\QGR}{\textsc{MAW-Qgrams}}
\def\dd{\mathinner{.\,.}}
\newcommand{\cO}{\mathcal{O}}
\newcommand{\SA}{\textsf{SA}}
\newcommand{\iSA}{\textsf{iSA}}

\newcommand{\LCP}{\textsf{LCP}}
\newcommand{\LCE}{\textsf{LCE}}
\newcommand{\lcp}{\textsf{lcp}}

\newcommand{\mina}{\mathcal{M}}
\newcommand{\LW}{\textsf{LW}}

 \newtheorem{theorem}{Theorem}
 
 \newtheorem{fact}[theorem]{Fact}

 \newtheorem{lemma}[theorem]{Lemma}
 \newtheorem{example}[theorem]{Example}
 
 \newtheorem{remark}[theorem]{Remark}

\begin{document}

\begin{frontmatter}

\title{Alignment-free sequence comparison using absent words\tnoteref{note1}}
\tnotetext[note1]{A preliminary version of this paper, without the first author, was presented at the 12th Latin American Theoretical Informatics Symposium (LATIN 2016)~\cite{Crochemore2016}.}
\author[a]{Panagiotis Charalampopoulos}
\ead{panagiotis.charalampopoulos@kcl.ac.uk}

\author[a]{Maxime Crochemore}
\ead{maxime.crochemore@kcl.ac.uk}

\author[b]{Gabriele Fici}
\ead{gabriele.fici@unipa.it}

\author[c]{Robert Merca\c{s}}
\ead{R.G.Mercas@lboro.ac.uk}

\author[a]{Solon P.\ Pissis}
\ead{solon.pissis@kcl.ac.uk}

\address[a]{Department of Informatics, King's College London, London, UK\\}
\address[b]{Dipartimento di Matematica e Informatica, Universit{\`a} di Palermo, Palermo, Italy\\}
\address[c]{Department of Computer Science, Loughborough University, UK\\}

\begin{abstract}
Sequence comparison is a prerequisite to virtually all comparative genomic analyses.
It is often realised by sequence alignment techniques, which are computationally expensive.
This has led to increased research into alignment-free techniques, which are based on measures referring to the composition
of sequences in terms of their constituent patterns. These measures, such as $q$-gram distance, are usually computed in time linear 
with respect to the length of the sequences. 
In this paper, we focus on the complementary idea: how two sequences can be efficiently compared based
on information that does not occur in the sequences.  
A word is an {\em absent word} of some sequence if it does not occur in the sequence.
An absent word is {\em minimal} if all its proper factors occur in the sequence.
Here we present the first linear-time and linear-space algorithm 
to compare two sequences by considering {\em all} their minimal absent words.
In the process, we present results of combinatorial interest, and also extend the proposed techniques to compare circular sequences.
We also present an algorithm that, given a word $x$ of length $n$, computes the largest integer for which all factors of $x$ of that length occur in some minimal absent word of $x$ in time and space $\cO(n)$. Finally, we show that the known asymptotic upper bound on the number of minimal absent words of a word is tight.
\end{abstract}

\begin{keyword}
sequence comparison \sep absent words \sep forbidden words \sep circular words \sep q-grams


\end{keyword}

\end{frontmatter}

\section{Introduction}

Sequence comparison is an important step in many basic tasks in bioinformatics, from phylogeny reconstruction to genome assembly. 
It is often realised by sequence alignment techniques, which are computationally expensive, often requiring quadratic time in the length of the sequences. 
This has led to increased research into \textit{alignment-free} techniques~\cite{Vinga01032003}. 
Hence standard notions for sequence comparison are gradually being complemented and in some cases replaced by alternative ones~\cite{Domazet-Loso:2009:EEP:1671627.1671629,WABI2015}. 
One such notion is based on comparing the words that are absent in each sequence~\cite{nullrly}. 
A word is an \textit{absent word} (or a forbidden word) of some sequence if it does not occur in the sequence. Absent words represent a type of \textit{negative information}: information about what does not occur in the sequence. 

Given a sequence of length $n$, the number of absent words of length at most $n$ is exponential in $n$. However, the number of certain classes of absent words is only linear in $n$.
This is the case for \textit{minimal absent words}, that is, absent words in the sequence for which all proper factors occur in the sequence~\cite{BeMiReSc00}. 
An upper bound on the number of minimal absent words of a word of length $n$ over an alphabet $\Sigma$ of size $\sigma$ is known to be $\sigma n$~\cite{Crochemore98automataand,Mignosi02}.
Hence it may be possible to compare sequences in time proportional to their lengths, for a fixed-sized alphabet, instead of proportional to the product of their lengths. 
In what follows, we mainly consider sequences over a {\em fixed-sized alphabet} since the most commonly studied alphabet in this context is $\{\texttt{A,C,G,T}\}$.

An $\cO(n)$-time and $\cO(n)$-space algorithm for computing all minimal absent words of a sequence of length $n$ over a fixed-sized alphabet based on the construction of suffix automata 
was presented in~\cite{Crochemore98automataand}.  
The computation of minimal absent words based on the construction of suffix arrays was considered in~\cite{Pinho2009}; although this algorithm has a linear-time performance in practice, 
the worst-case time complexity is $\cO(n^2)$. New $\cO(n)$-time and $\cO(n)$-space suffix-array-based algorithms were presented in~\cite{DBLP:conf/isit/FukaeOM12,MAW,PPAM2015} to bridge this unpleasant gap. 
An implementation of the algorithm presented in~\cite{MAW} is currently, to the best of our knowledge, the fastest available for the computation of minimal absent words.
A more space-efficient solution to compute all minimal absent words in time $\cO(n)$ was also presented in~\cite{Belazzougui2013} and an external-memory algorithm in~\cite{em-maw}.

In this paper, we consider the problem of comparing two sequences $x$ and $y$ of 
respective lengths $m$ and $n$, using their sets of minimal absent words.
In~\cite{Chairungsee2012109}, Chairungsee and Crochemore 
introduced a measure of similarity between two sequences based on the notion of minimal absent words. 
They made use of a length-weighted index to provide a measure of similarity between two sequences, using sample sets of their minimal absent words, by considering the length of
each member in the symmetric difference of these sample sets. This measure can be trivially computed in time and space $\cO(m + n)$ provided that these sample sets contain
minimal absent words of some bounded length $\ell$. For unbounded length, the same measure can be trivially computed in time $\cO(m^2 + n^2)$: for a given sequence, 
the cumulative length of all its minimal absent words can grow {\em quadratically} with respect to the length of the sequence. This length-weighted index forms the basis of a fundamentally new, recently introduced, algorithm for on-line pattern matching~\cite{pattern_matching}.

The same problem can be considered for two {\em circular} sequences. The measure of similarity of  Chairungsee and Crochemore can be used in this setting provided that one extends the 
definition of minimal absent words to circular sequences. In Section~\ref{sec:circ_seq_comp}, we give a definition of minimal absent words for a circular sequence from the Formal Language 
Theory point of view. We believe that this definition may also be of interest from the point of view of Symbolic Dynamics, which is the original context in which minimal absent words have been introduced~\cite{BeMiReSc00}.

We also find a connection between the information provided by minimal absent words and the information provided by the set of $q$-grams. The former (absent words) can be seen as some kind of \emph{negative} information, while the latter ($q$-grams) as \emph{positive} information.

\medskip
\noindent \textbf{Our contributions.} Here we make the following contributions:%
\begin{enumerate}
  \item[a)] We first show that the upper bound $\cO(\sigma n)$ on the number of minimal absent words of a word of length $n$ over an alphabet of size $\sigma$ is tight if $2 \leq \sigma \leq n$ (Section~\ref{sec:tightmaw}).
 \item[b)] We present an $\cO(m + n)$-time and $\cO(m + n)$-space algorithm to compute the similarity measure introduced by Chairungsee and Crochemore by considering {\em all} minimal absent words of two  sequences $x$ and $y$ of lengths $m$ and $n$, respectively, over a fixed-sized alphabet; thereby
 showing that it is indeed possible to compare two sequences in time proportional to their lengths (Section~\ref{sec:seq_comp}).
 \item[c)]  We show how this algorithm can be applied to compute this similarity measure for two circular sequences $x$ and $y$ of lengths $m$ and $n$, respectively, in the same time and space complexity as a result of the extension of the definition of minimal absent words to circular sequences (Section~\ref{sec:circ_seq_comp}).
  \item[d)] We then present an $\cO(n)$-time and $\cO(n)$-space algorithm that given a word $x$ of length $n$ over an integer alphabet computes the largest integer $q(x)$ for which each $q(x)$-gram of $x$ is a $q(x)$-gram of some minimal absent word of $x$ (Section~\ref{sec:q-grams}).
   \item[e)]  Finally, we provide an open-source code implementation of our algorithms for sequence comparison using minimal absent words and investigate potential applications of our theoretical findings (Section~\ref{sec:imp_app}).
\end{enumerate}

\section{Preliminaries}\label{sec:prem}

  We begin with basic definitions and notation.
  Let $y=y[0]y[1]\dd y[n-1]$ be a \textit{word} (or \textit{string}) of length $|y|=n$
over a finite ordered alphabet $\Sigma$ of size 
$|\Sigma|=\sigma=\cO(1)$. We also consider the case of words over an {\em integer alphabet}, where each letter is replaced by its rank in such a way that the resulting word consists of integers in the range $\{1,\ldots,n\}$. 

For two positions $i$ and $j$ on $y$, we denote by $y[i\dd j]=y[i]\dd y[j]$ the \textit{factor} 
(sometimes called \textit{substring}) of $y$ that 
starts at position $i$ and ends at position $j$ (it is of length $0$ if $j<i$), and by $\varepsilon$ 
the \textit{empty word}, word of length 0. 
  We recall that a prefix of $y$ is a factor that starts at position 0 
($y[0\dd j]$), a suffix is a factor that ends at position $n-1$ 
($y[i\dd n-1]$), and that a factor of $y$ is a \textit{proper} factor if 
it is not $y$ itself. A proper factor of $y$ that is neither a prefix nor a suffix of $y$ is called an {\em infix} of $y$. The set of all  factors of the word $y$ is denoted by $\F_y$. 
Any factor of length $q \geq 1$ of $y$ is called a {\em $q$-gram} (or {\em $q$-mer}) of $y$.
The {\em $q$-gram set} of $y$ is the set of all factors of length $q$ of $y$. 
We denote the reverse word of $y$ by $\textsf{rev}(y)$, i.e. $\textsf{rev}(y)=y[n-1]y[n-2]\ldots y[1]y[0]$. We say that a word $x$ is \emph{a power} of a word $y$ if there exists a positive integer $k$, $k>1$, such that $x$ is expressed as $k$ consecutive concatenations of $y$, denoted by $x=y^k$.

  Let $x$ be a word of length $m$ with $0<m\leq n$. 
  We say that there exists an \textit{occurrence} of $x$ in $y$, or, more 
simply, that $x$ \textit{occurs in} $y$, when $x$ is a factor of $y$.
  Every occurrence of $x$ can be characterised by a starting position in $y$. 
  We thus say that $x$ occurs at the \textit{starting position} $i$ in $y$ 
when $x=y[i \dd i + m - 1]$.
  Opposingly, we say that the word $x$ is an \textit{absent word} of
$y$ if it does not occur in $y$.
  The absent word $x$ of $y$ is \textit{minimal} if and only if all its proper factors 
occur in $y$. The set of all minimal absent words for a word $y$ is denoted by $\mina_y$. 
For example, if $y=\texttt{abaab}$, then $\mina_y=\{\texttt{aaa}, \texttt{aaba}, \texttt{bab}, \texttt{bb}\}$. In general, if we suppose that 
all the letters of the alphabet appear in $y$ which has length $n$, the length of a minimal absent word of $y$ lies between $2$ and $n+1$. 
It is equal to $n+1$ if and only if $y$ is of the form $a^n$ for some letter $a$. 
So, if $y$ contains occurrences of at least two different letters, the length of any minimal absent word of $y$ is upper bounded by $n$.

We now recall some basic facts about minimal absent words in Formal Language Theory. For further details and references the reader is recommended~\cite{fici}. A {\em language} over the alphabet $\Sigma$ is a set of finite words over $\Sigma$. A language is {\em regular} if it is recognised by a finite state automaton. A language is called {\em factorial} if it contains all the factors of its words, while it is called {\em antifactorial} if no word in the language is a proper factor of another word in the language. Given a word $x$, the language \emph{generated} by $x$ is the language $x^*=\{x^k\mid k\geq 0\}=\{\varepsilon, x, xx, xxx, \ldots\}$. The \emph{factorial closure} of a language $L$ is the language consisting of all factors of the words in $L$, that is, the language $\F_L=\cup_{y\in L} \F_y$.  
Given a factorial language $L$, one can define the (antifactorial) language of minimal absent words for $L$ as \[\mina_L=\{ aub\mid a,b \in \Sigma, u \in \Sigma^*,aub\notin L, au,ub\in L\}.\] 
Notice that $\mina_L$ is not the same language as the union of $\mina_x$ for $x\in L$. 
Every factorial language $L$ is uniquely determined by its (antifactorial) language of minimal absent words $\mina_L$, through the equation 
\begin{equation}\label{maw}
L=\Sigma^*\setminus \Sigma^*\mina_L\Sigma^*.
\end{equation}
The converse is also true, since by the definition of a minimal absent word we have \begin{equation}\label{maw2}
\mina_L=\Sigma L\cap L\Sigma \cap (\Sigma^*\setminus L).
\end{equation}
The previous equations define a bijection between factorial and antifactorial languages. Moreover, this bijection preserves regularity. In the case of a single  word $x$, the set of minimal absent words for $x$ is indeed the antifactorial language $\mina_{\F_{x}}$. Thus, applying \eqref{maw} and \eqref{maw2} to the language of factors of a single word, we have the following lemma.

\begin{lemma}\label{lem:bij}
Given two words $x$ and $y$, $x=y$ if and only if $\mina_x=\mina_y$.
\end{lemma}

Given a word $x$ of length $n$ over a fixed-sized alphabet, it is possible to compute a trie storing all the minimal absent words of $x$ in time and space linear in $n$. The size (number of nodes) of this trie is linear in $n$. Furthermore, we can retrieve $x$ from its set of minimal absent words in time and space linear in the size of the input trie representing the minimal absent words of $x$.  Indeed, the algorithm \textsc{MF-trie}, introduced in~\cite{Crochemore98automataand}, builds the tree-like deterministic automaton accepting the set of minimal absent words for a word $x$ taking as input the factor automaton of $x$, that is the minimal deterministic automaton recognising the set of factors of $x$. The leaves of the trie correspond to the minimal absent words for $x$, while the internal states are those of the factor automaton. Since the factor automaton of a word $x$ has less than $2|x|$ states (for details, see~\cite{CHL07}), this provides a representation of the minimal absent words of a word of length $n$ in space $\cO(n)$.

\subsection{Suffix array and suffix tree}
We denote by \SA{} the {\em suffix array} of a non-empty word $y$ of length $n$. \SA{} is an integer array of size $n$ storing the starting positions of all (lexicographically) sorted non-empty suffixes of $y$, i.e.~for all 
$1 \leq  r < n$ we have $y[\SA{}[r-1] \dd n-1] < y[\SA{}[r] \dd n - 1]$~\cite{SA}.
  Let \lcp{}$(r, s)$ denote the length of the longest common prefix between
$y[\SA{}[r] \dd n - 1]$ and $y[\SA{}[s] \dd n - 1]$ 
for all positions $r$, $s$ on $y$, and $0$ otherwise.
  We denote by \LCP{} the {\em longest common prefix} array of $y$ defined by 
\LCP{}$[r]=\lcp{}(r-1, r)$ for all $1 \leq r < n$, and 
\LCP{}$[0] = 0$. The inverse \iSA{} of the array \SA{} is defined by 
$\iSA{}[\SA{}[r]] = r$, for all $0 \leq r < n$. It is known that
  \SA{}~\cite{Nong:2009:LSA:1545013.1545570}, \iSA{}, and 
\LCP{}~\cite{indLCP} of a word of length $n$, over an integer alphabet, can be computed in time and space $\cO(n)$.

The \textit{suffix tree} $\mathcal{T}(y)$ of a non-empty word $y$ of length $n$ is a compact trie representing all suffixes of $y$. The nodes of the trie which become nodes of the suffix
tree are called {\em explicit} nodes, while the other nodes are called {\em implicit}. Each edge
of the suffix tree can be viewed as an upward maximal path of implicit nodes starting with an explicit node. Moreover, each node belongs to a unique path of that kind. Thus, each node of the trie can be represented in the suffix tree by the edge it belongs to and an index within the corresponding path.
We let  $\mathcal{L}(v)$  denote the \textit{path-label} of a node $v$, i.e., the concatenation of the edge labels along the path from the root to $v$. We say that $v$ is  path-labelled  $\mathcal{L}(v)$. Additionally, $\mathcal{D}(v)= |\mathcal{L}(v)|$ is used to denote  the \textit{string-depth} of node $v$. Node  $v$ is  a \textit{terminal} node if its path label is a suffix of $y$, that is, $\mathcal{L}(v) = y[i \dd n-1]$ for some $0 \leq i < n$; here $v$ is also labelled with index $i$. It should be clear that each  factor of $y$ is uniquely represented by either an explicit or an implicit node of $\mathcal{T}(y)$. The \textit{suffix-link} of a node $v$ with path-label $\mathcal{L}(v)= \alpha w$ is a pointer to the node path-labelled $w$, where  $\alpha \in \Sigma$ is a single letter and $w$ is a word. The suffix-link of $v$ is defined if $v$ is an explicit node of $\mathcal{T}(y)$, different from the root.
In any standard implementation of the suffix tree, we assume that each node of the suffix tree is able to access its parent. Note that once $\mathcal{T}(y)$ is constructed, it can be traversed in a depth-first manner to compute the string-depth $\mathcal{D}(v)$ for each node $v$. Let $u$ be the parent of $v$. Then the string-depth $\mathcal{D}(v)$ is computed by adding  $\mathcal{D}(u)$ to the length of the label of edge $(u,v)$. If $v$ is the root, then $\mathcal{D}(v) = 0$.
It is known that the suffix tree of a word of length $n$, over an integer alphabet, can be computed in time and space $\cO(n)$~\cite{farach1997optimal}.

\subsection{A measure of similarity between words based on minimal absent words}
In what follows, as already proposed in~\cite{MAW}, for every word $y$, the set of minimal absent words 
associated with $y$, denoted by $\mina_y$, is represented as a set of tuples $\langle a, i,j \rangle$, 
where the corresponding minimal absent word $x$ of $y$ is defined by
$x[0]=a$, $a \in \Sigma$, and $x[1 \dd m-1] = y[i \dd j]$, where $j-i+1=m \geq 2$.
It is known that if $|y|=n$ and $|\Sigma|=\sigma$, then $|\mina_y| \leq \sigma n$~\cite{Mignosi02}.

In~\cite{Chairungsee2012109}, Chairungsee and Crochemore introduced a measure of similarity between two words $x$ and $y$ based on the notion of minimal absent words. 
Let $\mina_x^\ell$ (respectively $\mina_y^\ell$) denote the set of minimal absent words of length at most $\ell$ of $x$ (respectively $y$).
The authors made use of a length-weighted index to provide a measure of similarity between $x$ and $y$, using their sample sets $\mina_x^\ell$ and $\mina_y^\ell$, by considering the length of
each member in the symmetric difference $(\mina_x^\ell \bigtriangleup \mina_y^\ell)$ of the sample sets. For sample sets $\mina_x^\ell$ and $\mina_y^\ell$, they defined this index to be
$$\LW_\ell(x,y) = \sum_{w \in \mina_x^\ell \bigtriangleup \mina_y^\ell} \frac{1}{|w|^2}.$$
In this paper we consider a more general measure of similarity for two words $x$ and $y$. It is based on the set $\mina_x \bigtriangleup \mina_y$, and is defined by
$$\LW(x,y) = \sum_{w \in \mina_x \bigtriangleup \mina_y} \frac{1}{|w|^2},$$
so without any restriction on the lengths of minimal absent words.
The smaller the value of $\LW(x,y)$, the more similar we assume $x$ and $y$ to be. Note that $\LW(x,y)$ is affected by both the cardinality of $\mina_x \bigtriangleup \mina_y$ and the lengths of its elements; longer words in $\mina_x \bigtriangleup \mina_y$ contribute less in the value of $\LW(x,y)$ than shorter ones. Hence, intuitively, the shorter the words in $\mina_x \bigtriangleup \mina_y$, the more dissimilar $x$ and $y$ are.

We provide the following examples for illustration. Let $x=\texttt{abaab}$ and $y=\texttt{aabbbaa}$. We have $\mina_x = \{ \texttt{aaa}, \texttt{aaba}, \texttt{bab}, \texttt{bb} \}$ and $\mina_y = \{ \texttt{aaa}, \texttt{bbbb}, \texttt{aba}, \texttt{abba}, \texttt{bab}, \texttt{baab} \}$. Thus, 
$$\mina_x \bigtriangleup \mina_y = \{ \texttt{aaba}, \texttt{aba}, \texttt{abba}, \texttt{baab}, \texttt{bb}, \texttt{bbbb} \},$$ so that
\[
\LW(\texttt{abaab},\texttt{aabbbaa})=4\cdot \frac{1}{4^2}+\frac{1}{3^2}+\frac{1}{2^2}=\frac{11}{18}.
\]
Similarly, 
\[
\LW(\texttt{aaa},\texttt{bbb})=\frac{17}{8}
\]
and
\[
\LW(\texttt{aaa},\texttt{aaaa})=\frac{41}{400}.
\]

This measure of similarity aims at quantifying the distance between two words by means of their minimal absent words. In fact, we show here that this measure is consistent with the notion of {\em distance} (metric) in the mathematical sense.

\begin{lemma}\label{lem:metric}
$\LW(x,y)$ is a metric on $\Sigma^*$.
\end{lemma}
\begin{proof}
It is clear that {\em non-negativity}, $\LW(x,y) \geq 0$  for any $x,y \in \Sigma^*$, and {\em symmetry}, $\LW(x,y)=\LW(y,x)$  for any $x,y \in \Sigma^*$, are satisfied. In addition, we have from Lemma~\ref{lem:bij} that $\LW(x,y)=0$ if and only if $x=y$, hence {\em identity} is also satisfied. Furthermore, given three sets $A$, $B$ and $C$, we have by the properties of the symmetric difference that $A \bigtriangleup B \subseteq (A \bigtriangleup C) \cup (C \bigtriangleup B)$. Thus, given three words $x$, $y$ and $z$, for every $w \in \mina_x \bigtriangleup \mina_y$ we have that $w \in \mina_x \bigtriangleup \mina_z$ or $w \in \mina_z \bigtriangleup \mina_y$ and therefore $w$ contributes by $1/|w|^2$ to $\LW(x,y)$ and to one of $\LW(x,z)$ and $\LW(z,y)$. This shows that $\LW(x,y) \leq \LW(x,z)+\LW(z,y)$ for any $x,y,z \in \Sigma^*$ and so {\em triangle inequality} is also satisfied.
\end{proof}

Based on this similarity measure we consider the following problem:

\defproblem{\MAW}{a word $x$ of length $m$ and a word $y$ of length $n$}{$\LW(x,y)$.}

In  Section~\ref{sec:seq_comp}, we show that this problem can be solved  in $\cO(m + n)$-time and space.
 
\subsection{Extension to circular words}


We also consider the aforementioned problem for two circular words. A circular word of length $m$ can be viewed as a traditional linear word which has the left- and right-most letters 
wrapped around. Under this notion, the same circular word can be seen as $m$ different 
linear words, which would all be considered equivalent. More formally, given a word $x$ of length $m$, we denote 
by $x^{\langle i \rangle}=x[i \dd m-1]x[0 \dd i-1]$, $0 \leq i < m$, the $i$-th \textit{rotation} of $x$, where $x^{\langle 0 \rangle}=x$. Given two words $x$ and $y$, we define $x\sim y$ if  there exists $i$, $0 \leq i < |x|$, such that $y=x^{\langle i \rangle}$. A \emph{circular word} $\tilde{x}$ is a conjugacy class of the equivalence relation $\sim$. Given a circular word $\tilde{x}$, any (linear) word $x$ in the equivalence class $\tilde{x}$ is called a \emph{linearisation} of the circular word $\tilde{x}$. Conversely, given a linear word $x$, we say that $\tilde{x}$ is a \emph{circularisation} of $x$ if $x$ is a linearisation of $\tilde{x}$.

The factorial closures of the languages generated by two rotations of the same word $x^{\langle i \rangle}$ and $x^{\langle j \rangle}$, i.e.~the languages $\F_{(x^{\langle i \rangle})^*}$ and $\F_{(x^{\langle j \rangle})^*}$, coincide, so one can unambiguously define the (infinite) language $\F_{\tilde{x}}$ of factors of the circular word $\tilde{x}$  as the language $\F_{x^*}$, where $x$ is any linearisation of $\tilde{x}$. This is coherent with the fact that a circular word can be seen as a word drawn on a circle, where there is no beginning and no end.

In Section~\ref{sec:circ_seq_comp}, we give the definition of the set $\mina_{\tilde{x}}$ of minimal absent words for a circular word $\tilde{x}$. 
We prove that the following problem can be solved within the same time and space complexity as its counterpart in the linear case.

\defproblem{\CMAW}
{a word $x$ of length $m$ and a word $y$ of length $n$}{$\LW(\tilde{x},\tilde{y})$, 
where $\tilde{x}$ and $\tilde{y}$ are circularisations of $x$ and $y$, respectively.}

\subsection{Minimal absent words and $q$-grams}
In Section~\ref{sec:q-grams}, we present an $\cO(n)$-time and $\cO(n)$-space algorithm that given a word $x$ of length $n$ computes $q(x)$, the largest integer for which each $q(x)$-gram of $x$ is a $q(x)$-gram of some minimal absent word of $x$.

\defproblem{\QGR}{a word $x$ of length $n$}{$q(x)$}

\section{Tight asymptotic bound on the number of minimal absent words}\label{sec:tightmaw}

An important property of the minimal absent words of a word $x$, that is at the basis of the algorithms presented in next sections, is that their number is linear in the size of $x$.  Let $x$ be a word of length $n$ over an alphabet of size $\sigma$. In~\cite{Mignosi02} it is shown that the total number of minimal absent words of $x$ is smaller than or equal to $\sigma n$. In the following lemma we show that $\mathcal{O}(\sigma n)$ is a tight asymptotic bound for the number of minimal absent words of $x$ whenever $2 \leq \sigma \leq n$.

\begin{lemma}\label{tightbound}
The upper bound $\cO(\sigma n)$ on the number of minimal absent words of a word $x$ of length $n$ over an alphabet of size $\sigma$ is tight if $2 \leq \sigma \leq n$.
\end{lemma}
\begin{proof}
The total number of minimal absent words of $x$ is smaller than or equal to $\sigma n$~\cite{Mignosi02}. Hence $\cO(\sigma n)$ is an asymptotic upper bound for the number of minimal absent words of $x$. In what follows we provide examples to show that this bound is tight if $2 \leq \sigma \leq n$.

Let $\Sigma=\{a_1, a_2\}$, i.e. $\sigma=2$, and consider the word $x=a_2 a_1^{n-2} a_2$ of length $n$. All words of the form $a_2 a_1^k a_2$, for $0 \leq k \leq n-3$, are minimal absent words of $x$. Hence $x$ has at least $n-2=\Omega (n)$ minimal absent words.

Let $\Sigma=\{a_1,a_2,a_3,\ldots,a_\sigma\}$ with $3 \leq \sigma \leq n$ and consider the word $x=a_2 a_1^k a_3 a_1^k \ldots a_i a_1^k a_{i+1} \ldots a_{\sigma} a_1^k a_1^m$, where $k=\lfloor \frac{n}{\sigma-1}\rfloor -1$ and $m=n-(\sigma-1)(k+1)$. Note that $|x|=n$. Further note that $a_i a_1^j$ is a factor of $x$, for all $2 \leq i \leq \sigma $ and $0 \leq j \leq k$. Similarly, $a_1^j a_l$ is a factor of $x$, for all $3 \leq l \leq \sigma $ and $0 \leq j \leq k$. Thus, all proper factors of all the words in the set $S=\{ a_i a_1^j a_l \: | \: 0 \leq j \leq k, \: 2 \leq i \leq \sigma, \: 3 \leq l \leq \sigma \}$ occur in $x$. However, the only words in $S$ that occur in $x$ are the ones of the form $a_i a_1^k a_{i+1}$, for $2 \leq i < \sigma$. Hence $x$ has at least $(\sigma-1)(\sigma-2)(k+1)-(\sigma-2)=(\sigma-1)(\sigma-2)\lfloor \frac{n}{\sigma-1}\rfloor-(\sigma-2)=\Omega(\sigma n)$ minimal absent words.
\end{proof}

\section{Sequence comparison using minimal absent words}\label{sec:seq_comp}

The goal of this section is to provide the first linear-time and linear-space algorithm for computing the similarity measure (see Section~\ref{sec:prem}) between two words defined over a fixed-sized alphabet. 
To this end, we consider two words $x$ and $y$ of lengths $m$ and $n$, respectively, and their associated sets of minimal absent words, $\mina_x$ and $\mina_y$, respectively. 
Next, we give a linear-time and linear-space solution for the {\MAW} problem. 

It is known from~\cite{Crochemore98automataand} and~\cite{MAW} that we can compute the sets $\mina_x$ and $\mina_y$ in linear time and space relative to the two lengths $m$ and $n$, respectively.
The idea of our strategy consists of a merge sort on the sets $\mina_x$ and $\mina_y$, after they have been ordered with the help of suffix arrays.
To this end, we construct the suffix array associated to the word $w=xy$, together with the implicit $\LCP$ array corresponding to it. 
All of these structures can be constructed in time and space $\cO(m+n)$, as mentioned earlier. Furthermore, we can preprocess the array \textsf{LCP} for range minimum queries, which we denote by $\textsf{RMQ}_\textsf{LCP}$~\cite{Fischer11}.
With the preprocessing complete, the longest common prefix $\LCE$ of two suffixes of $w$ starting at positions $p$ and $q$ can be computed in
constant time~\cite{LCE}, using the formula
\[\LCE(w,p,q)=\textsf{LCP}[\textsf{RMQ}_{\textsf{LCP}}(\textsf{iSA}[p]+1,\textsf{iSA}[q])].\]

Using these data structures, it is straightforward to sort the tuples in the sets $\mina_x$ and $\mina_y$ lexicographically. 
That is, two tuples, $x_1$ and $x_2$, are ordered such that the one being the prefix of the other comes first, or according to the letter following their longest common prefix, when the former is not the case.
In our setting, the latter is always the case since $\mina_x$ is prefix-free by the definition of minimal absent words. To do this, we simply go once through the suffix array associated with $w$ and assign to each tuple in $\mina_x$, respectively $\mina_y$, the rank of the suffix starting at the position indicated 
by its second component, in the suffix array. 
Since sorting an array of $n$ distinct integers, such that each is in $[0,n-1]$, can be done in time $\cO(n)$ (using for example bucket sort) we can sort each of the sets of minimal absent words, taking into consideration the letter on the first position and these ranks. Thus, from now on, we assume that $\mina_x=\{x_0, x_1,\ldots, x_k\}$ where $x_i$ is  lexicographically smaller than $x_{i+1}$, for $0\leq i <k\leq \sigma m$, and $\mina_y=\{y_0, y_1,\ldots, y_\ell\}$, where $y_j$ is  lexicographically smaller than $y_{j+1}$, for $0\leq j <\ell\leq \sigma n$.

We now proceed with the merge. Thus, considering that we are analysing $x_{i}$ from $\mina_x$ and $y_{j}$ from $\mina_y$, we note that the two are equal if and only if $x_i[0]=y_j[0]$ and 
$$\LCE(w,x_i[1], |x|+y_j[1])\geq \ell, \mbox{ where } \ell=x_i[2]-x_i[1]=y_j[2]-y_j[1].$$ 
In other words, the two minimal absent words are equal if and only if their first letters coincide, they have equal length $\ell+1$, and the longest common prefix of the suffixes of $w$ starting at the positions indicated by 
the second components of the tuples has length at least $\ell$.

Such a strategy will empower us with the means for constructing a new set $\mina_{x,y}=\mina_x\cup\mina_y$. At each step, when analysing tuples $x_i$ and $y_j$, we proceed as follows:
$$
\mina_{x,y} = \left\{
                \begin{array}{l c l}
                  \mina_{x,y}\cup \{x_i\}, & \mbox{and increment } i, &\qquad\mbox{if } x_i < y_j;\\
                  \mina_{x,y}\cup \{y_j\}, & \mbox{and increment } j, &\qquad\mbox{if } x_i > y_j;\\
                  \mina_{x,y}\cup \{x_i=y_j\}, &   \mbox{and increment both } i \mbox{ and } j, &\qquad \mbox{if } x_i = y_j.
                \end{array}
              \right.
$$
Observe that the last condition is saying that basically each common tuple is added only once to their union.

Furthermore, simultaneously with this construction we can also calculate the similarity between the words, given by $\LW(x,y)$, which is initially set to $0$. 
Thus, at each step, when comparing the tuples $x_i$ and $y_j$, we update 
$$
\LW(x,y) = \left\{
                \begin{array}{l c l}
                  \LW(x,y) + \frac{1}{|x_i|^2}, & \mbox{and increment } i, &\mbox{if } x_i < y_j;\\
                  \LW(x,y) + \frac{1}{|y_j|^2}, & \mbox{and increment } j, &\mbox{if } x_i > y_j;\\
                  \LW(x,y), &  \mbox{and increment both } i \mbox{ and } j, &\mbox{if } x_i = y_j.
                \end{array}
              \right.
$$
We impose the increment of both $i$ and $j$ in the case of equality as in this case we only look at the symmetric difference between the sets of minimal absent words.

As all these operations take constant time and we perform them once per each tuple in $\mina_x$ and $\mina_y$, it is easily concluded that the whole operation takes, in the case of a fixed-sized alphabet, time and space $\cO(m+n)$. 
Thus, we can compute the symmetric difference between the {\em complete} sets of minimal absent words, as opposed to~\cite{Chairungsee2012109}, of two words defined over a fixed-sized alphabet, in linear time and space with respect to the lengths of the two words. We hence obtain the following result:

\begin{theorem}
\label{the:maw}
Problem \MAW{} can be solved in time and space $\cO(m+n)$.
\end{theorem}

\section{Circular sequence comparison using minimal absent words}\label{sec:circ_seq_comp}


In this section we extend the notion of minimal absent words to circular words. 
Recall from Section \ref{sec:prem} that, given a circular word $\tilde{x}$, the set $\F_{\tilde{x}}$ of factors of $\tilde{x}$ is defined as the (infinite) set $\F_{x^*}$, where $x$ is any linearisation of $x$. We therefore define the set $\mina_{\tilde{x}}$ of minimal absent words of the circular word $\tilde{x}$ as the set of minimal absent words of the language $\F_{x^*}$, where $x$ is any linearisation of $x$.

For instance, let $x=\texttt{aabbabb}$.  Then we have
$$\mina_{\tilde{x}}=\{\texttt{aaa}, \texttt{aba}, \texttt{bbb}, \texttt{aabbaa}, \texttt{babbab}\}.$$ 

Although $\F_{x^*}$ is an infinite language, the set $\mina_{\tilde{x}}=\mina_{\F_{x^*}}$ of minimal absent words of $\tilde{x}$ is always finite. More precisely, we have the following structural lemma (see also \cite{FiReRi17}).

\begin{lemma}\label{lem:stru}
  Let $\tilde{x}$ be a circular word and $x$ any linearisation of $\tilde{x}$. Then 
\begin{equation}\label{eq:cmf}
   \mina_{\tilde{x}}=\mina_{xx}^{|x|}.
\end{equation}
That is, the minimal absent words of the circular word $\tilde{x}$ are precisely the minimal absent words of the (linear) word $xx$ whose length is not greater than the length of $x$, where $x$ is any linearisation of $\tilde{x}$.
\end{lemma}

\begin{proof}
If $aub$, with $a,b\in \Sigma$ and $u\in \Sigma^*$, is an element in $\mina_{xx}^{|x|}$, then clearly  $aub\in \mina_{\F_{x^*}}=\mina_{\tilde{x}}$.

Conversely, let $aub$, with $a,b\in \Sigma$ and $u\in \Sigma^*$, be an element in $\mina_{\tilde{x}}=\mina_{\F_{x^*}}$. Then $aub\notin \F_{x^*}$, while $au,ub\in \F_{x^*}$. So, there exists a letter $\bar{b}$ different from $b$ such that $au\bar{b}\in \F_{x^*}$ and a letter $\bar{a}$ different from $a$ such that $\bar{a}ub\in \F_{x^*}$. Therefore, $au,\bar{a}u,ub,u\bar{b}\in \F_{x^*}$.
Any word of length at least $|x|-1$ cannot be extended to the right nor to the left by different letters in $\F_{x^*}$ as such factors would yield two rotations of $x$ with different letter multiplicities. Hence $|aub|\leq |x|$. Since $au$ and $ub$ are factors of some rotation of $x$, we have $au,ub\in \F_{xx}$, whence $aub\in \mina_{xx}$.
\end{proof}

The equality (\ref{eq:cmf}) was first introduced as the definition of the set of minimal absent words of a circular word in~\cite{6979851}. 

Recall that a word $x$ is \emph{a power} of a word $y$ if there exists a positive integer $k$, $k>1$, such that $x$ is expressed as $k$ consecutive concatenations of $y$, denoted by $x=y^k$. 
Conversely, a word $x$ is {\em primitive} if $x=y^k$ implies $k=1$. Notice that a word is primitive if and only if any of its rotations is. We can therefore extend the definition of primitivity to circular words. The definition of $\mina_{\tilde{x}}$ does not allow one to uniquely reconstruct $\tilde{x}$ from $\mina_{\tilde{x}}$, unless $\tilde{x}$ is known to be primitive, since it is readily verified that $\F_{x^*}=\F_{xx^*}$ and therefore also the minimal absent words of these two languages coincide. However, from the algorithmic point of view, 
this issue can be easily managed by storing the length $|x|$ of a linearisation $x$ of $\tilde{x}$ together with the set $\mina_{\F_{x^*}}$.
Moreover, in most practical scenarios, for example when dealing with biological sequences, it is highly unlikely that the input circular word is not primitive.

Using the result of Lemma~\ref{lem:stru}, we can easily extend the algorithm described in the previous section to the case of circular words. 
That is, given two circular words $\tilde{x}$ of length $m$ and  $\tilde{y}$ of length $n$, we can compute in time and space $\cO(m+n)$ 
the distance $\LW(\tilde{x},\tilde{y})$. We hence obtain the following result.

\begin{theorem}
\label{the:cmaw}
Problem \CMAW{} can be solved in time and space $\cO(m+n)$.
\end{theorem}

\section{From minimal absent words to $q$-grams}\label{sec:q-grams}

In this section we consider a word $x$ of length $n$ over an integer alphabet. Our aim is to provide a measure of the extent to which some {\em positive} information about $x$, the $q$-gram sets of $x$, exist unaltered in the set of minimal absent words of $x$, which can be seen as {\em negative} information about $x$. More specifically, we define $q(x)$ as the largest integer for which each $q(x)$-gram of $x$ is a $q(x)$-gram of some minimal absent word of $x$. Note that, for instance, the set of $q$-grams of $x$ is used in molecular biology applications such as genome assembly~\cite{Pevzner14082001}.

\begin{example}
 Consider the word $x=\texttt{abaab}$ over the alphabet $\Sigma=\{ \texttt{a},\texttt{b} \}$. Its set of minimal absent words is $\mina_x=\{\texttt{aaa}, \texttt{aaba}, \texttt{bab}, \texttt{bb}\}$. The set of $2$-grams of $x$ is $\{\texttt{aa}, \texttt{ab}, \texttt{ba} \}$ and, as can be easily seen, each of them is a factor of some word in $\mina_x$. The set of $3$-grams of $x$ is $\{ \texttt{aab}, \texttt{aba}, \texttt{baa} \}$ and we observe that $\texttt{baa}$ is not a factor of any of the words in $\mina_x$. We can hence conclude that in this case $q(x)=2$.
\end{example}
We present a non-trivial $\cO(n)$-time and $\cO(n)$-space algorithm to compute $q(x)$.

\subsection{Useful properties}

Let $h(x)$ be the length of a shortest factor of $x$ that occurs only once in $x$. In addition, let $t(x)$ be the length of a shortest infix (factor that is not a prefix nor a suffix) of $x$ that occurs only once in $x$.

Following the proof of~\cite[Proposition~10]{FICI2006214}, any factor of a word $x$ that occurs more than once in $x$ is a factor of some minimal absent word of $x$ and hence $q(x) \geq h(x)-1$.

\begin{lemma}\label{lem:bound}
For any word $x$ it holds that $q(x) \leq t(x)+1$.
\end{lemma}

\begin{proof}
Consider any non-empty infix $u$ of $x$ that occurs only once and suppose it is preceded by letter $a$ and followed by letter $b$. Then $aub$ can not be a factor of any of the minimal absent words of $x$ as the largest infix of any minimal absent word of $x$ must occur at least twice in $x$, once in an occurrence of the largest proper prefix of this minimal absent word in $x$ and once in an occurrence of its larger proper suffix in $x$. Note that $au \neq ub$, since otherwise $aub=a^{2+|u|}$ and then $u$ does not occur only once in $x$. It thus follows that $q(x) \leq t(x)+1$.
\end{proof}

\begin{fact}
We can compute $h(x)$ and $t(x)$ --- and hence obtain the relevant bounds for $q(x)$ --- in time $\cO(n)$ for a word $x$ of length $n$.
\end{fact}

\begin{remark}
The relation between minimum unique substrings and maximum repeats
 has been investigated in~\cite{DBLP:journals/fuin/IlieS11}.
\end{remark}

Note that all 1-grams $a_i$ that occur in $x$ are trivially contained in some minimal absent word of the form $a_i^k$ for some $k$, so in what follows we assume that the factors of $x$, for which we want to examine when they are factors of some minimal absent word of $x$, are of the form $aub$, where $a$ and $b$ are (not necessarily distinct) letters and $u$ a (possibly empty) word. It is clear that any such factor $aub$ of $x$ occurring only once can not be a minimal absent word itself. In addition, following the proof of Lemma~\ref{lem:bound}, it can not be an infix of a minimal absent word. In the following lemma we provide a necessary and sufficient condition for $aub$ to occur as a prefix of some minimal absent word of $x$.

\begin{lemma}\label{lem:equiv}
Let $aub$, with $a,b \in \Sigma$ and $u$ a word, be a factor occurring only once in a word $x$. The two following statements are equivalent:

\begin{enumerate}
\item\label{cond1}
$aub$ is a prefix of some minimal absent word of $x$;
\item\label{cond2}
$ub$ occurs at least twice in $x$ and if $j_1<j_2<\ldots<j_k$ are the starting positions of its occurrences, with $j_m-1$ being the starting position of the occurrence of $aub$, at least one of the following holds: (i) $x[j_m+k] \neq x[j_i+k]$ for some $i$, $k$ such that $j_m+k \leq n-1$ and $j_i+k \leq n-1$; (ii) $m \neq 1$.
\end{enumerate}
\end{lemma}
\begin{proof}
($\ref{cond1}. \Rightarrow~\ref{cond2}.$): Consider a word $aub$ occurring just once in $x$ and appearing as a prefix of some minimal absent word. Firstly $aub$ is not itself a minimal absent word, so any minimal absent word that has $aub$ as a prefix must be of the form $aubvd$, where $v$ is a possibly empty word and $d$ a letter. The existence of this minimal absent word means that $ubvd$ occurs in $x$ and it is not preceded by $a$ (so $ub$ occurs at least twice in $x$) and that $aubv$ occurs in $x$ and either: 
\begin{itemize}
\item[($i$)]
it is followed by a letter $c \neq d$, or 
\item[($ii$)]
$aubv$ is a suffix of $x$.
\end{itemize}

($\ref{cond1}. \Leftarrow~\ref{cond2}.$): If (i) holds, then for any minimal such $k$ we have that $x[j_m-1\dd j_m+k-1]x[j_i+k]$ is a minimal absent word.\\
If (i) does not hold, but (ii) holds, then $x[j_m\dd n-1]x[j_1+n-j_m]$ is a minimal absent word.
\end{proof}

Similarly, whether $aub$ is a suffix of some minimal absent word of $x$ depends on the extensions of $\textsf{rev}(u)a$ in $\textsf{rev}(x)$.
\bigskip

\subsection{Computing $q(x)$}

In this section we present Algorithm \Algo{MawToQgrams} that, given a word $x$ of length $n$ over an integer alphabet, computes $q(x)$ in time and space $\mathcal{O}(n)$. The algorithm first creates the suffix trees of $x$ and $\textsf{rev}(x)$ and then preprocesses them in time $\mathcal{O}(n)$. The preprocessing phase for each tree is a depth-first search traversal which allows us to store in each node $v$ a boolean variable $\mathcal{B}(v)$ which indicates if there is any branching in the subtree rooted at $v$ and a variable $\mathcal{S}(v)$ indicating the starting position of the first occurrence of $\mathcal{L}(v)$ in $x$. The latter can be done in time $\mathcal{O}(n)$ since we store the starting position of the suffix corresponding to each terminal node while constructing the suffix tree. Algorithm \Algo{MawToQgrams} then calls Routines \Algo{InfixBound}, \Algo{PrefixBound} and \Algo{SuffixBound} to compute $q(x)$.

\bigskip
\begin{algo}{MawToQgrams}{x}
\SET {\mathcal{T}(x)}{\Call{SuffixTree}{x}}
\SET {\mathcal{T}(\textsf{rev}(x))}{\Call{SuffixTree}{\textsf{rev}(x)}}
\DOFOR{\mbox{each node $v \in \mathcal{T}(x)$}}
\SET{\mathcal{B}(v)}{\mbox{$\texttt{true}$ if there is any branching below $v$ and $\texttt{false}$ otherwise}}
\SET{\mathcal{S}(v)}{\mbox{Starting position of the first occurrence of $\mathcal{L}(v)$ in $x$}}
\OD
\DOFOR{\mbox{each node $v \in \mathcal{T}(\textsf{rev}(x))$}}
\SET{\mathcal{B}(v)}{\mbox{$\texttt{true}$ if there is any branching below $v$ and $\texttt{false}$ otherwise}}
\SET{\mathcal{S}(v)}{\mbox{Starting position of the first occurrence of $\mathcal{L}(v)$ in $x$}}
\OD
\SET{q}{\Call{InfixBound}{x}}
\SET{q}{\Call{PrefixBound}{x, q}}
\SET{q}{\Call{SuffixBound}{x, q}}
\RETURN{q}
\end{algo}

\bigskip
As we have already seen, all the factors of $x$ that occur more than once in $x$ also occur in some minimal absent word of $x$. Hence our aim is to identify a shortest factor of $x$ that is not a factor of any of the minimal absent words of $x$.

We first present Routine \Algo{Test} that, given as inputs $i$ and $j$, tests if the factor $x[i\dd j]$ of $x$ that occurs only once in $x$ also occurs in some minimal absent word of $x$. Let $x[i\dd j]=aub$, where $a,b \in \Sigma$ and $u$ is a word. The routine first checks if $x[i\dd j]$ occurs as a prefix of some minimal absent word of $x$ by checking statement~(\ref{cond2}) of Lemma~\ref{lem:equiv} as follows. It considers the node of $\mathcal{T}(x)$ with path-label $x[i+1\dd j]=ub$; note that in the pseudocode this node is denoted by $\textsc{Node}(\mathcal{T}(x))(i+1, j)$. If this node is explicit, then it is named $v$, while if it is implicit, then the destination of the edge it is on is named $v$. The routine then checks in time $\cO(1)$ if $\mathcal{B}(v)$ is $\texttt{true}$ or if $\mathcal{S}(v) \leq i$. If this is the case, then $aub$ is a factor of some minimal absent word and the test returns $\texttt{true}$.
Otherwise, the analogous check is performed for $\textsf{rev}(x)[n-j\dd n-i-1]=\textsf{rev}(u)a$ in $\mathcal{T}(\textsf{rev}(x))$.
If both checks are unsuccessful, then the routine returns $\texttt{false}$. We discuss how to efficiently obtain the desired nodes later in this section.

\bigskip
\begin{algo}{Test}{$$i$, $j$$}
\SET{v}{\textsc{Node}(\mathcal{T}(x))(i+1, j)}
\IF{\Call{IsImplicit}{v}}
\SET{(v_1,v_2)}{\Call{Edge}{v}}
\SET{v}{v_2}
\FI
\IF{\mbox{$\mathcal{B}(v)= \texttt{true} \ or \ \mathcal{S}(v) \leq i$}}
\RETURN{\texttt{true}}
\FI
\SET{v}{\textsc{Node}(\mathcal{T}(\textsf{rev}(x)))(n-j, n-i-1)}
\IF{\Call{IsImplicit}{v}}
\SET{(v_1,v_2)}{\Call{Edge}{v}}
\SET{v}{v_2}
\FI
\IF{\mbox{$\mathcal{B}(v)=\texttt{true} \ or \ \mathcal{S}(v) \leq n-j-1$}}
\RETURN{\texttt{true}}
\FI
\RETURN{\texttt{false}}
\end{algo}

\bigskip
Now note that the factors of $x$ that occur only once in $x$ are the labels of the leaves and of the implicit nodes on the edges between internal nodes and leaves in the suffix tree. Hence, if node $u$ is a leaf with $\mathcal{L}(u)=x[i\dd n-1]$, then $x[i\dd \mathcal{D}(parent(u))+1]$ corresponds to the shortest unique factor of $x$ occurring at $i$. We can thus find $t(x)$ and all the infixes of $x$ of a given length that occur only once in $x$ in time $\mathcal{O}(n)$. We can also obtain the shortest unique prefix and the shortest unique suffix of $x$ in time $\cO(1)$.

Routine \Algo{InfixBound} first computes all unique infixes of $x$ of length $t(x)$ and tests if there is any of them that does not occur in any minimal absent word of $x$, in which case we have that $q(x) \leq t(x)-1$. If this is not the case, the routine computes all unique infixes of $x$ of length $t(x)+1$ and tests if there is any of them that does not occur in any minimal absent word of $x$, in which case we have that $q(x) \leq t(x)$. Otherwise, we use the bound $q(x) \leq t(x)+1$ shown in Lemma~\ref{lem:bound}, and hence do not have to increment again.

\bigskip
\begin{algo}{InfixBound}{x}
\SET{\mbox{$\ell$}}{t(x)}
\SET{\mathcal{I}}{\mbox{all $(i,j)$ such that $x[i\dd j]$ is a unique infix of length $\ell$}}
\DOFOR{\mbox{each  $(i,j) \in \mathcal{I}$}}
\IF{\Call{Test}{i,j}=\texttt{false}}
	\RETURN{\mbox{$\ell-1$}}
\FI
\OD
\SET{\mathcal{I'}}{\mbox{all $(i,j)$ such that $x[i\dd j]$ is a unique infix of length $\ell+1$}}
\DOFOR{\mbox{each  $(i,j) \in \mathcal{I'}$}}
\IF{\Call{Test}{i,j}=\texttt{false}}
	\RETURN{\mbox{$\ell$}}
\FI
\OD
\RETURN{\mbox{$\ell+1$}}
\end{algo}

\bigskip
Finally, we also perform the same test for the prefixes and suffixes of $x$ that occur only once and their length is smaller than the bound we have at that point. This is done by Routines \Algo{PrefixBound} and \Algo{SuffixBound}. We can then conclude on the value of $q(x)$.

\bigskip
\begin{algo}{PrefixBound}{$$x$, $q$$}
\SET{p}{\mbox{Length of shortest unique prefix of $x$}}
\DOWHILE{p \leq q}
\IF{\Call{Test}{0,p-1}=\texttt{false}}
	\RETURN{p-1}
\FI
\SET{p}{p+1}
\OD
\RETURN{q}
\end{algo}

\bigskip
\begin{algo}{SuffixBound}{$$x$, $q$$}
\SET{s}{\mbox{Length of shortest unique suffix of $x$}}
\DOWHILE{s \leq q}
\IF{\Call{Test}{n-s,n-1}=\texttt{false}}
	\RETURN{s-1}
\FI
\SET{s}{s+1}
\OD
\RETURN{q}
\end{algo}

\bigskip

We now discuss how to answer the queries $\textsc{Node}(\mathcal{T}(x))(i+1, j)$ in line $1$ of \Algo{Test} in time $\cO(n)$ in total. We first discuss how to answer the queries asked within Routine \Algo{InfixBound}. While computing sets $\mathcal{I}$ and $\mathcal{I'}$, alongside the pair $(i,j)$, we also store a pointer to the deepest explicit ancestor $v_{i,j}$ of the node with path-label $x[i \dd j]$. We can do this in time $\cO(n)$ due to how we compute $\mathcal{I}$ and $\mathcal{I'}$. We have that $\mathcal{D}(v_{i,j})=t(x)-1=j-i$ for $(i,j) \in \mathcal{I}$ and $\mathcal{D}(v_{i,j})=t(x)-1=j-i-1$ or $\mathcal{D}(v_{i,j})=t(x)=j-i$ for $(i,j) \in \mathcal{I'}$. Following the suffix-link from the explicit node $v_{i,j}$ we retrieve the node with path-label $x[i+1 \dd j-1]$ or the node with path-label $x[i+1 \dd j-2]$. We then only need to answer at most two child queries for each such node to obtain the node with path-label $x[i+1 \dd j]$. Answering such queries on-line bears the cost of $\cO(\log \sigma)$ per query for integer alphabets or that of non-determinism if we make use of perfect hashing to store the edges at every node of $\mathcal{T}(x)$~\cite{DBLP:journals/jacm/FredmanKS84}. We instead answer these queries off-line: it is well-known that we can answer $q$ child queries off-line during a depth-first traversal of the suffix tree in $\cO(n+q)$ deterministic time by first sorting the queries at each node of $\mathcal{T}(x)$. We first answer one child query per pair $(i,j)$ in a batch and then the potential second ones in another batch. The total time required for this is $\cO(n)$.
Routine \Algo{PrefixBound} only considers nodes with path-labels of the form $x[1 \dd h]$, which can be found by following the edges upwards from the node with path-label $x[1 \dd n-1]$. Routine \Algo{SuffixBound} only considers terminal nodes to which we can afford to store pointers while creating $\mathcal{T}(x)$. We answer the respective queries for $\mathcal{T}(\textsf{rev}(x))$ (line $7$ of \Algo{Test}) in a similar fashion. Finally, having the pointers to the required nodes, we perform all the tests off-line.

Alternatively, we can obtain a deterministic $\cO(n)$-time solution by employing a data structure for a special case of Union-Find~\cite{DBLP:journals/jcss/GabowT85} --- a detailed description of this technique can be found in the appendix of~\cite{zstringsv1}.

\begin{theorem}\label{the:q-grams-maws}
Problem \QGR{} can be solved in time and space $\mathcal{O}(n)$.
\end{theorem}
\begin{proof}
We build and preprocess the suffix trees of $x$ and $\textsf{rev}(x)$ in time and space $\mathcal{O}(n)$~\cite{farach1997optimal}. Based on Lemma~\ref{lem:bound} we then have to perform the test for $\mathcal{O}(n)$ factors, which we can find in time $\mathcal{O}(n)$. The tests are performed in total time $\mathcal{O}(n)$ by finding the required nodes and using the preprocessed suffix trees to check statement~(\ref{cond2}) of Lemma~\ref{lem:equiv}. We only need extra space $\mathcal{O}(n)$ to store a representation of the computed factors and perform the tests.
\end{proof}

\section{Implementation and applications}\label{sec:imp_app}

We implemented the algorithms presented in Section~\ref{sec:seq_comp} and Section~\ref{sec:circ_seq_comp} as programme \textsf{scMAW} 
to perform pairwise sequence comparison for a set of sequences using minimal absent words.
\textsf{scMAW} uses programme \textsf{MAW}~\cite{MAW} for linear-time and linear-space computation of minimal absent words using suffix array. 
\textsf{scMAW} was implemented in the $\textsf{C}$ programming language and developed under GNU/Linux operating system. 
It takes, as input argument, a file in MultiFASTA format with the input sequences, and then any of the two methods, for {\em linear} or {\em circular} sequence comparison, 
can be applied. It then produces a file in PHYLIP format with the distance matrix as output.
Cell $[x,y]$ of the matrix stores $\LW(x,y)$ (or $\LW(\tilde{x},\tilde{y})$ for the circular case).
The implementation is distributed under the GNU General Public License (GPL), and it is available at \url{http://github.com/solonas13/maw}, 
which is set up for maintaining the source code and the man-page documentation. 
Notice that {\em all} input datasets and the produced outputs referred to in this section are publicly maintained at the same web-site.

An important feature of the proposed algorithms is that they require space linear in the length of the sequences 
(see Theorem~\ref{the:maw} and Theorem~\ref{the:cmaw}). Hence, we were also able to implement \textsf{scMAW}
using the Open Multi-Processing (OpenMP) PI for shared memory multiprocessing programming to distribute the workload 
across the available processing threads without a large memory footprint.

\subsection{Applications}
Recently, there has been a number of studies on the biological significance of absent words in various species~\cite{nullrly,minabpro,Silva02042015}.
In~\cite{minabpro}, the authors presented dendrograms from dinucleotide relative abundances in sets of minimal absent words for prokaryotes and eukaryotic genomes.
The analyses support the hypothesis that minimal absent words are inherited through a common ancestor, in addition to lineage-specific inheritance, 
only in vertebrates. Very recently, in~\cite{Silva02042015}, it was shown that there exist three minimal words in the Ebola virus genomes which are absent from human genome.
The authors suggest that the identification of such species-specific sequences may prove to be useful for the development of both diagnosis and therapeutics.

In this section, we show a potential application of our results for the construction of dendrograms for DNA sequences with circular structure.
Circular DNA sequences can be found in viruses, as plasmids in archaea and bacteria, and in the mitochondria and plastids of eukaryotic cells.
Circular sequence comparison thus finds applications in several contexts such as reconstructing phylogenies using viroids RNA~\cite{conf/gcb/MosigHS06} or Mitochondrial DNA (MtDNA)~\cite{mtDNA-phylo}. 
Conventional tools to align circular sequences could yield an incorrectly high genetic distance between closely-related species. Indeed, when sequencing
molecules, the position where a circular sequence starts can be totally arbitrary. Due to this {\em arbitrariness}, a suitable rotation of one sequence would give much better results 
for a pairwise alignment~\cite{SEA2015,WABI2015}. In what follows, we demonstrate the power of minimal absent words to pave a path to resolve 
this issue by applying Lemma~\ref{lem:stru} and Theorem~\ref{the:cmaw}.
Next we do not claim that a solid phylogenetic analysis is presented but rather an investigation for potential applications of our theoretical findings.

We performed the following experiment with synthetic data. 
First, we simulated a basic dataset of DNA sequences using INDELible~\cite{indelible}.
The number of taxa, denoted by $\alpha$, was set to $12$; 
the length of the sequence generated at the root of the tree, denoted by $\beta$, was set to 2500bp;
and the substitution rate, denoted by $\gamma$, was set to $0.05$. 
We also used the following parameters: a deletion rate, denoted by $\delta$, of $0.06$ \textit{relative} to substitution rate of $1$; 
and an insertion rate, denoted by $\epsilon$, of $0.04$ \textit{relative} to substitution rate of $1$. 
The parameters were chosen based on the genetic diversity standard measures observed for sets of MtDNA sequences from primates and mammals~\cite{SEA2015}. 
We generated another instance of the basic dataset, containing one {\em arbitrary} rotation of each of the $\alpha$ sequences from the basic dataset. 
We then used this randomised dataset as input to \textsf{scMAW} by considering $\LW(\tilde{x},\tilde{y})$ as the distance metric. 
The output of \textsf{scMAW} was passed as input to \textsf{NINJA}~\cite{ninja}, an efficient implementation of
neighbour-joining~\cite{NJ}, a well-established hierarchical clustering algorithm for inferring dendrograms (trees). 
We thus used \textsf{NINJA} to infer the respective tree $T_1$ under the neighbour-joining criterion.
We also inferred the tree $T_2$ by following the same pipeline, but by considering $\LW(x,y)$ as distance metric, 
as well as the tree $T_3$ by using the {\em basic} dataset as input of this pipeline and $\LW(\tilde{x},\tilde{y})$ as distance metric.
Hence, notice that $T_3$ represents the original tree. 
Finally, we computed the pairwise Robinson-Foulds (RF) distance~\cite{RFdistance} between: $T_1$ and $T_3$; and $T_2$ and $T_3$.

Let us define \textit{accuracy} as the difference between 1 and the relative pairwise RF distance.
We repeated this experiment by simulating different datasets $<\alpha,\beta,\gamma,\delta,\epsilon>$ and measured the corresponding accuracy.  
The results in Table~\ref{tab:accuracy} (see $T_1$ vs. $T_3$) suggest that by considering $\LW(\tilde{x},\tilde{y})$ we can always 
re-construct the original tree even if the sequences have first been arbitrarily rotated (Lemma~\ref{lem:stru}). 
This is not the case (see $T_2$ vs. $T_3$) if we consider $\LW(x,y)$. Notice that $100\%$ accuracy denotes a (relative) pairwise RF distance of 0.
\begin{table}[!t]
\begin{center}
\scalebox{0.7}{
\begin{tabular}{l|c|c} \hline
Dataset $<\alpha,\beta,\gamma,\delta,\epsilon>$	& $T_1$ vs. $T_3$ &  $T_2$ vs. $T_3$ \\ \hline
$<12,2500,0.05,0.06,0.04>$		& 100\%		& 100\%\\
$<12,2500,0.20,0.06,0.04>$		& 100\%		& 88,88\%\\
$<12,2500,0.35,0.06,0.04>$		& 100\%		& 100\%\\
$<25,2500,0.05,0.06,0.04>$		& 100\%		& 100\%\\
$<25,2500,0.20,0.06,0.04>$		& 100\%		& 100\%\\
$<25,2500,0.35,0.06,0.04>$		& 100\%		& 100\%\\
$<50,2500,0.05,0.06,0.04>$		& 100\%		& 97,87\%\\
$<50,2500,0.20,0.06,0.04>$		& 100\% 	& 97,87\%\\
$<50,2500,0.35,0.06,0.04>$		& 100\%		& 100\%\\ \hline
\end{tabular}
}
\end{center}
\caption{Accuracy measurements based on relative pairwise RF distance}
\label{tab:accuracy}
\end{table}

\section{Final remarks}

In this paper, complementary to measures that refer to 
the composition of sequences in terms of their constituent patterns, we considered sequence comparison using
minimal absent words, information about what does not occur in the sequences.
We presented the first linear-time and linear-space algorithm to compare two sequences by considering {\em all} their minimal absent words.
In the process, we presented some results of combinatorial interest, and also extended the proposed techniques to circular sequences.
The power of minimal absent words is highlighted by the fact that they provide a tool for sequence comparison that is as efficient for circular as it is for linear 
sequences; whereas this is not the case, for instance, using the general edit distance model~\cite{Maes}.
In addition, we presented a linear-time and linear-space algorithm that given a word $x$ computes the largest integer $q(x)$ for which each $q(x)$-gram of $x$ is a $q(x)$-gram of some minimal absent word of $x$.
Finally, a preliminary experimental study shows the potential of our theoretical findings with regards to alignment-free sequence comparison using negative information.

\section*{Acknowledgements}

We warmly thank Alice H\'{e}liou (\'{E}cole Polytechnique) for her inestimable code contribution and Antonio Restivo (Universit{\`a} di Palermo) for useful discussions. We also thank the anonymous reviewers for their constructive comments which greatly improved the presentation of the paper. Gabriele Fici's work was supported by the PRIN 2010/2011 project ``Automi e Linguaggi Formali: Aspetti Matematici e Applicativi'' of the Italian Ministry of Education (MIUR) and by the ``National Group for Algebraic and Geometric Structures, and their Applications'' (GNSAGA -- INdAM).
Robert Merca{\c s}'s work was supported by a Newton Fellowship of the Royal Society.
Solon  P.~Pissis's  work was  supported  by  a  Research  Grant (\#RG130720) awarded by the Royal Society.

\bibliographystyle{elsarticle-num}
\bibliography{references}

\end{document}